\documentclass{article}
\pdfoutput=1
\usepackage{ijcai25}

\usepackage{times}
\usepackage{soul}
\usepackage{url}
\usepackage[hidelinks]{hyperref}
\usepackage[utf8]{inputenc}
\usepackage[small]{caption}
\usepackage{graphicx}
\usepackage{amsmath}
\usepackage{amsthm}
\usepackage{booktabs}
\usepackage{algorithm}
\usepackage[noend]{algorithmic}
\usepackage[switch]{lineno}
\usepackage{amsmath}
\usepackage{amssymb}
\usepackage{subfigure} 

\newtheorem{example}{Example}
\newtheorem{theorem}{Theorem}
\newtheorem{definition}{Definition}
\newtheorem{problem}{Problem Defnition}
\newtheorem{lemma}{Lemma}

\newtheorem*{paragraph*}{Paragraph}
\newtheorem{corollary}{Corollary}

\title{Effective Index Construction Algorithm for Optimal $(k,\eta)$-cores Computation}

\author{
Shengli Sun$^1$
\and
Peng Xu$^1$\and
Guanming Jiang$^1$\and
Philip S.Yu$^2$\and
Yi Li$^3$\\
\affiliations
$^1$School of Software and Microelectronics, Peking University, China\\
$^2$Department of Computer Science, University of Illinois at Chicago, Chicago\\
$^3$Mckelvey Engineering School, Washington University, St. Louis\\
\emails
slsun@ss.pku.edu.cn,
\{peng.xu, jiangguanming\}@stu.pku.edu.cn,
psyu@uic.edu, lyi1@wustl.edu
}

\begin{document}

\maketitle

\begin{abstract}
    Computing $(k,\eta)$-cores from uncertain graphs is a fundamental problem in uncertain graph analysis. UCF-Index is the state-of-the-art resolution to support $(k,\eta)$-core queries, allowing the $(k,\eta)$-core for any combination of $k$ and $\eta$ to be computed in an optimal time. However, this index constructed by current algorithm is usually incorrect. During decomposition, the key is to obtain the $k$-probabilities of its neighbors when the vertex with minimum $k$-probability is deleted. Current method uses recursive floating-point division to update it, which can lead to serious errors. We propose a correct and efficient index construction algorithm to address this issue. Firstly, we propose tight bounds on the $k$-probabilities of the vertices that need to be updated, and the accurate $k$-probabilities are recalculated in an on-demand manner. Secondly, vertices partitioning and progressive refinement strategy is devised to search the vertex with the minimum $k$-probability, thereby reducing initialization overhead for each $k$ and avoiding unnecessary recalculations. Finally, extensive experiments demonstrate the efficiency and scalability of our approach.
\end{abstract}

\section{Introduction}
On uncertain graphs, $(k,\eta)$-core mining is a crucial method for complex network analysis in real-world environments that are often filled with noise \cite{aggarwal2010managing} and measurement errors \cite{adar2007managing}. Given an uncertain graph $\mathcal{G}$, a $(k,\eta)$-core is a maximal subgraph where each vertex has a degree at least \textit{k} with probability no less than $\eta$, \text{where} $\eta \in [0,1]$. The $(k,\eta)$-core is widely used for graph analysis tasks, such as community search \cite{luo2021efficient,miao2022reliable}, fraud detection \cite{wu2021containment}, influence maximization \cite{bonchi2014core,dai2021core,calio2020cores}, and task-driven team formation \cite{esfahani2019efficient,yang2019index}, etc. 

UCF-Index is the state-of-the-art resolution to support efficient computation of $(k,\eta)$-cores for any possible values of $k$ and $\eta$ with an optimal time complexity that is linear to the number of vertices in the result set  \cite{yang2019index}. The $\eta$-threshold of some vertex $u$ is the maximum value of $\eta$ such that a $(k,\eta)$-core containing $u$ exists. UCF-Index efficiently supports online search of $(k,\eta)$-cores by maintaining $\eta$-threshold of each vertex for each $k$. However, the UCF-Index constructed by current algorithm is usually incorrect.

The $k$-probability of vertex $u$ on $\mathcal{G}$ is the probability that $u$ has a degree no less than $k$. During $(k,\eta)$-core decomposition, the key is to obtain the $k$-probability of its neighbors when a vertex is deleted. Yang et al. iteratively removed the vertex with minimum $k$-probability to calculate $\eta$-threshold, and sought to update its neighbors' $k$-probability by a recursive floating-point number division operation. Consequently, the error of the neighbor’s $\eta$-threshold so obtained can be up to $O(\varepsilon /{{(1-{{p}_{e}})}^{k}})$, where $\varepsilon$ is the error of the floating-point number representation of a computer and ${p}_{e}$ is the probability of the edge to be deleted \cite{dai2021core}. As an illustrative example shown in Figure \ref{fig:error_UCF}, suppose that the floating-point number precision of a computer is $10^{-3}$ and $k=2$, $\eta=0.107$. The correct (2,0.107)-core is $\{v_5,v_6,v_7\}$ as shown in Figure\ref{fig:right_example}, while the (2,0.107)-core obtained by current UCF-Index construction algorithm produces an incorrect result $\{v_1,v_2,v_3,v_5,v_6,v_7\}$ as shown in Figure\ref{fig:error_example}.


To accurately compute $(k,\eta)$-core with a specific parameter set $k$ and $\eta$, Dai et al. proposed to recalculate the $k$-probability using dynamic programming (DP) to avoid division calculation errors \cite{dai2021core}. Nevertheless, for UCF-Index, which can quickly support $(k,\eta)$-core queries with any combinations of parameters $k$ and $\eta$, the question of how to correctly and efficiently build it remains unresolved.


\begin{figure}[tbp]
\centering
\subfigure[Accurate results]{
    \label{fig:right_example} 
    \includegraphics[width=0.21\textwidth]{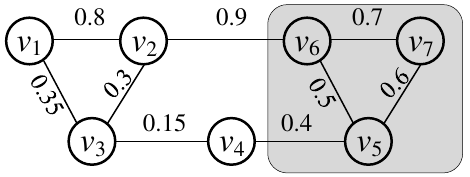}
}
\hfill
\subfigure[Inaccurate results]{
    \label{fig:error_example} 
    \includegraphics[width=0.22\textwidth]{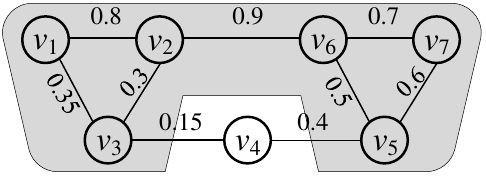}
}
\vspace{-1em}
\caption{The $(k,\eta)$-core of $\mathcal{G}$ for $k=2$ and $\eta=0.107$}
\vspace{-1em}
\label{fig:error_UCF}
\end{figure}



In this paper, we propose an efficient algorithm to build a correct UCF-Index. Our basic strategy is to recalculate the neighbors' $k$-probability, hence correctness of the index can be guaranteed. Our major idea is to minimize the number of recalculations based on a bounding and lazy refreshing strategy, and to limit the search space for vertex with minimum $k$-probability based on vertices partitioning and a progressive refinement strategy. The principal contributions of this paper are summarized as follows:
\begin{itemize}
    \item We propose tight bounds to estimate the $k$-probability of the vertices that need to be updated, and propose a lazy refreshing strategy where the recalculation of $k$-probability is carried out on-demand.
    \item We partition vertices into different layers based on the property of $\eta$-threshold. The search for the vertex with the minimum $k$-probability is conducted layer by layer, thereby reducing the overhead of $k$-probability initialization in each iteration and avoiding unnecessary recalculations.
    \item Extensive experiments demonstrate the efficiency and scalability of our approach. Our proposed algorithm \textit{OptiUCF} outperforms the baseline algorithm by one to two orders of magnitude.
\end{itemize}

\section{Preliminary}
\subsection{Core Decomposition on Uncertain Graphs}
Given an uncertain graph $\mathcal{G}(V,E,p)$, where $V$ denotes the vertex set, $E$ denotes the edge set, and $p$ is a function that assigns each edge $e \in E$ to a probability value in [0,1]. We denote the probability of an edge $e$ as $p_e$. Let $N_u(\mathcal{G})$ be the set of neighbors of $u$ in $\mathcal{G}$ and $deg(u,\mathcal{G})=|N_u(\mathcal{G})|$ be the degree of $u$. An induced graph of $\mathcal{G}$ based on a vertex set $C\subseteq V$ is defined as $\mathcal{G}[C]=(V_C,E_C,p)$, where $V_C=C$ and $E_C=\{(u,v)|u,v\in C \land (u,v)\in E\}$. In line with existing works \cite{potamias2010k,jin2011distance,li2014efficient}, we assume that the existing probability of each edge is independent in uncertain graphs. Under this assumption, the well-known possible world semantics can be applied to the analysis of uncertain graphs \cite{bonchi2014core}. Thus, there are $2^{|E|}$ possible graph instances in an uncertain graph $\mathcal{G}$. The probability of a possible graph instance $G$ of $\mathcal{G}$, denoted as $\Pr (G)$, can be calculated as follows:
\begin{align}
    \Pr (G)=\prod\limits_{e\in {{E}_{G}}}{{{p}_{e}}}\prod\limits_{e\in E\setminus{{E}_{G}}}{(1-{{p}_{e}})}
\end{align}%

Let $\mathcal{G}_{u}^{\ge k}$  denote the set of all possible instances of $\mathcal{G}$ where $u$ has a degree of at least $k$. We have the following equation \cite{bonchi2014core}:

\begin{equation}
      \Pr (deg (u,\mathcal{G})\ge k)=\sum\limits_{G\in \mathcal{G}_{u}^{\ge k}}{\Pr (G)}
\label{equation:G}
\end{equation}



The preconception of $(k,\eta)$-core is $k$-core, a well-known cohesive subgraph model on the deterministic graph. $K$-core is defined as the maximal subgraph of graph $G$ such that all vertices in the subgraph have a degree of at least $k$.  Based on the possible world semantics, the concept of $(k,\eta)$-core is defined as follows.

\begin{definition}[$(k,\eta)$-core]
    Given an uncertain graph $\mathcal{G}$ and a probability threshold $\eta \in [0,1]$, the $(k,\eta)$-core is a maximal connected induced subgraph $\mathcal{G}'=(V',E',p)$ in which the probability that each vertex has a degree no less than $k$ is no less than $\eta$, i.e., $\forall{u \in V'}, \Pr(deg(u,\mathcal{G}') \ge k) \ge \eta$.
\label{definition: kn_core}
\end{definition}

\begin{example}
    Consider the uncertain graph $\mathcal{G}$ in Figure \ref{fig:3n_core}. We denote the induced subgraph $\mathcal{G}[\{v_1,v_2,v_3,v_4\}]$ as $\mathcal{G'}$. Given an integer $k=2$, we have $\Pr(deg(v_1,\mathcal{G'}) \ge 2) = 0.80$, $\Pr(deg(v_2,\mathcal{G'}) \ge 2) = \Pr(deg(v_4,\mathcal{G'}) \ge 2) = 0.85$, 
 and $\Pr(deg(v_3,\mathcal{G'}) \ge 2) = 0.98804$. $\mathcal{G'}$ is maximal. Therefore, $\mathcal{G'}$ is a $(2,0.80)$-core.
\end{example}




\subsection{UCF-Index}
To support the online search for $(k,\eta)$-cores on large-scale uncertain graphs, Yang et al. proposed a forest-based index UCF-Index, which is composed of $k$ $\eta$-trees. Before introducing the UCF-Index, we first provide the definition of the $\eta$-threshold.

\begin{definition}[$\eta$-threshold]
    Given an uncertain graph $\mathcal{G}$ and an integer $k$, the $\eta$-threshold of a vertex $u$, denoted by $\eta$-$thres_k(u)$, is the maximum value of $\eta$ such that $u$ can be contained in a $(k,\eta)$-core.
\label{definition:eta}
\end{definition}

Based on Definition \ref{definition:eta}, a vertex $u$ is in the $(k,\eta)$-core if and only if $\eta$-$thres_k(u) \ge \eta$ \cite{yang2019index}. 

\begin{figure}[tbp]
\centering
\subfigure[The $(k,\eta)$-core of $\mathcal{G}$ for $k=2$ and $\eta=0.8$]{
    \label{fig:3n_core} 
    \includegraphics[width=0.22\textwidth]{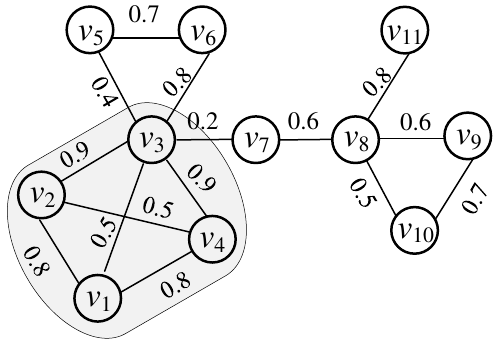}
}
\hfill
\subfigure[The $\eta$-$\text{tree}_2$ of $\mathcal{G}$]{
    \label{fig:2_tree} 
    \includegraphics[width=0.22\textwidth]{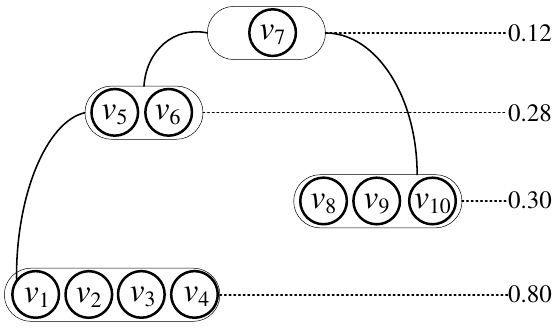}
}
\vspace{-1em}
\caption{The $(k,\eta)$-core and $\eta$-$\text{tree}_2$ of $\mathcal{G}$}
\vspace{-1em}
\label{fig:core_Sa}
\end{figure}

UCF-Index uses a tree structure to maintain the $\eta$-threshold of each vertex for each $k$, denoted by $\eta$-tree$_k$. Figure \ref{fig:2_tree} shows the $\eta$-tree of $\mathcal{G}$ for $k=2$. The $\eta$-threshold of each vertex is listed on the right. Given $\eta=0.28$, we can first find the node $R=\{v_5,v_6\}$ such that the $\eta$-thresholds of all vertices in $R$ are not less than $\eta$. Then, the induced subgraph of the set of all vertices $\{v_1,v_2,v_3,$$v_4,v_5,v_6\}$ in the subtree rooted at $R$ is a $(2,0.28)$-core.

Next, we introduce the current method for constructing the UCF-Index based on $k$-probability.

\begin{definition}[$k$-probability]
    Given an uncertain graph $\mathcal{G}$ and an integer $k$, the $k$-probability of a vertex $u$ in  $\mathcal{G}$, denoted by $k$-$prob(u,\mathcal{G})$, is the probability that $\Pr(deg(u,\mathcal{G}) \ge k)$.
\end{definition}




$K$-probability can be computed by a dynamic programming (DP) algorithm based on the following equation:
\begin{equation}
    \Pr (deg (u,\mathcal{G})\ge k)=1-\sum\limits_{i=0}^{k-1}{\Pr (deg (u,\mathcal{G})=i)}
\label{equation:computeG}
\end{equation}
Let $E(u,\mathcal{G})=\{e_1,e_2,...,e_{deg(u,\mathcal{G})}\}$ be the set of edges incident to $u$ in $\mathcal{G}$ and $E_h(u,\mathcal{G})=\{e_1,e_2,...,e_h\}$ be a subset of $E(u,\mathcal{G})$, where $h\in [1,deg (u,\mathcal{G})]$. Let $\mathcal{G}_h=(V,E\ \backslash\ (E(u,\mathcal{G})\ \backslash \ E_h(u,\mathcal{G})),p)$ be a subgraph of $\mathcal{G}$ that excludes the edges in $E(u,\mathcal{G})\backslash E_h(u,\mathcal{G})$. We use ${{X}_{u}}(h,i)=\Pr (deg (u,\mathcal{G}_h)=i)$ to represent the probability of a vertex $u$ that has a degree $i$ in $\mathcal{G}_h$, where $i\in [0,h]$. Therefore, the DP equation is as follows \cite{bonchi2014core}:
\begin{equation}
    {{X}_{u}}(h,i)={{p}_{{{e}_{h}}}}{{X}_{u}}(h-1,i-1)+(1-{{p}_{{{e}_{h}}}}){{X}_{u}}(h-1,i)
\label{con:dp}
\end{equation}

The initial states of the DP equation are as follows:
\begin{equation}
\left\{
\begin{array}{ll}
  X_u(0, 0) = 1, \\ 
  X_u(h, -1) = 0, \quad h \in [0, deg(u, \mathcal{G})] \\ 
  X_u(h, j) = 0, \quad h \in [0, deg(u, \mathcal{G})], \, j \in [h+1, i] \\ 
\end{array}
\right.
\label{con:dpinit}
\end{equation}

In peeling stage, we use $\mathcal{G}_{\lnot e}=(V,E\backslash e,p)$ to represent the subgraph of $\mathcal{G}$ after deleting $e$. When an edge $e$ is removed from $E(u,\mathcal{G})$, the $k$-probability of $u$ can be updated by: 
\begin{equation}
\resizebox{0.91\linewidth}{!}{%
$
\begin{aligned}
    & \Pr (deg (u,\mathcal{G}_{\lnot e})=i) \\
    & = \frac{\Pr (deg (u,\mathcal{G})=i) - p_{e} \Pr (deg (u,\mathcal{G}_{\lnot e})=i-1)}{1-p_{e}}
\end{aligned} $
}
\label{con:update}
\end{equation}
where $i\in [1,k]$, and if $i=0$, there is $\Pr (deg (u,\mathcal{G}_{\lnot e})=0)=\frac{1}{1-p_e}\cdot \Pr(deg(u,\mathcal{G})=0)$.


However, as analyzed by \cite{dai2021core}, the floating-point division operation of Equation (\ref{con:update}) will cause serious errors. During the process of UCF-Index construction, the vertex with minimum $k$-probability is iteratively removed and its neighbors' $k$-probability is updated accordingly by Equation (\ref{con:update}). As a result, the UCF-Index constructed in this way is incorrect.


\begin{problem}
     Given an uncertain graph $\mathcal{G}$, we aim to efficiently and accurately construct the UCF-Index to support the computation of $(k,\eta)$-cores in $\mathcal{G}$ for all possible integer $k$ and $\eta \in [0,1]$.
\end{problem}

\section{Baseline Algorithm to Build a Correct UCF-Index}
\label{subsec:bc}

Our basic strategy is to use Equation (\ref{con:dp}) to invoke DP to recompute the $k$-probabilities of the affected vertices, thereby avoiding division operations and ensuring the correctness of the $k$-probabilities. Each invocation of DP to recompute the $k$-probability is referred to as a \textit{refreshing} in the rest of this article. Therefore, for any vertex $u \in V$, the complexity of refreshing is $O(deg (u,\mathcal{G})\cdot k)$.


We present the pseudo-code of the baseline algorithm in Algorithm \ref{alg:base}. We decrease $k$ from $k_{max}$ to 1 in each iteration, where $k_{max}$ is the largest $k$ that enables the $k$-core of $\mathcal{G}$ to be non-empty. We calculate the $k$-core of $\mathcal{G}$ and compute $k$-$prob(u,\mathcal{G}')$ of each vertex $u \in V'$ (Lines 2-3). Next, we iteratively remove vertex $p_{crt}$, the current vertex with minimum $k$-probability, and add it into stack $S$. Then, we get $\eta$-$thres_k(p_{crt})$ and use DP to recompute the $k$-probabilities for all neighbors of $p_{crt}$ (Lines 5-10). 
Finally, we build the $\eta$-tree based on $S$ and their $\eta$-thresholds \cite{yang2019index}.


\begin{algorithm}[tb]
    \caption{Baseline}
    \label{alg:base}
    \textbf{Input}: An uncertain graph $\mathcal{G}(V,E,p)$\\
    \textbf{Output}: Correct UCF-Index of $\mathcal{G}$
    \begin{algorithmic}[1] 
        \FOR{ $k \leftarrow k_{max}$ to 1}
        \STATE $\mathcal{G}'(V',E',p) \leftarrow$ the $k$-core of $\mathcal{G}$
        \STATE Compute $k$-$prob(u,\mathcal{G}')$ for all the vertices $u \in V'$
        \STATE $S \leftarrow$ initialize an empty stack; \textit{curThres}$\leftarrow 0$
        \WHILE{$\mathcal{G}'$ is not empty}
        \STATE $p_{crt}\leftarrow \arg {{\min }_{u\in V'}}k$-$prob(u,\mathcal{G}')$
        \STATE \textit{curThres}$\leftarrow \max(k$-$prob(p_{crt},\mathcal{G}')$,\textit{curThres})
        \STATE $\eta$-$thres_k(p_{crt}) 
        \leftarrow$\textit{curThres}; $S$.push$(p_{crt})$
        \STATE $V' \leftarrow V' \backslash \{p_{crt}\}$
        \STATE Compute $k$-$probs$ using DP for all neighbors of {\it $p_{crt}$}
        \ENDWHILE
        \STATE Construct $\eta$-tree$_k$
        \ENDFOR
        \STATE \textbf{return} $\eta$-tree$_k$ for all $1\le k\le {{k}_{\max }}$
    \end{algorithmic}
\end{algorithm}

\noindent 
\paragraph{Optimization Road Map of Baseline.} We do not use floating-point division in the Baseline algorithm, which ensures the correctness of the UCF-Index. However, the $k$-probability of a vertex may be repeatedly calculated due to the deletion of its neighbors in peeling phase. Repeatedly calling a high-overhead DP process is inefficient. Moreover, in each iteration, the Baseline algorithm needs to initialize the \( k \)-probabilities of all vertices and search the entire uncertain graph to obtain the vertex with minimum $k$-probability, resulting in a high time cost. Therefore, we attempt to improve the Baseline from the following two aspects. Firstly, we propose bounds to estimate the $k$-probability of the vertices that need to be recalculated, and propose a lazy refreshing strategy. Secondly, vertices partitioning and progressive refinement strategy is devised to search for the vertex with the minimum $k$-probability. This reduces the initialization overhead for each $k$ and eliminates unnecessary refreshing.


\section{Effective UCF-Index Construction}
In this section, we propose an effective UCF-Index construction method. First, we develop a lazy refreshing strategy to reduce redundant recalculations. Then, the progressive refinement strategy is used to further accelerate index construction. Finally, we present the details of our \textit{OptiUCF} algorithm.

\subsection{Lazy Refreshing Strategy}
\label{subsec:optimal}
The general idea is to delay and minimize recalculation of $k$-probability based on tight bounds. We first introduce a basic vertices deletion strategy.


\begin{definition}[Associative Vertex]
    Given an uncertain graph $\mathcal{G}$, if the deletion of vertex $u$ causes $deg(v,\mathcal{G})<k$, then vertex $v$ is an associative vertex of $u$.
\end{definition}

In Figure \ref{fig:3n_core}, let $k=2$, since $deg(v_6,\mathcal{G})<k$ after deleting $v_5$, $v_6$ is an \textit{associative vertex} of $v_5$ and $\eta$-$thres_k(v_6)$ is equal to $\eta$-$thres_k(v_5)$. Thus, when removing $u$ from $\mathcal{G}$, we can delete all \textit{associative vertices} of $u$ without refreshing.

Next, we introduce two bounds to support the estimation of $k$-probability.

\noindent
\paragraph{The Beta-Function Based Bound.} Let  $d_u=|N_u(\mathcal{G})|$, $p_{max}(u)$ and $p_{min}(u)$ denote the minimum and maximum probability on the edges incident to $u$, respectively, i.e., $p_{min}(u)=\min(\{p_e|e\in E(u,\mathcal{G})\})$,\ $p_{max}(u)=\max(\{p_e|e\in E(u,\mathcal{G})\})$. Based on the inequalities proved by Bonchi et al \cite{bonchi2014core}, the upper bound and lower bound of $k$-$prob(u,\mathcal{G})$ for each vertex $u \in \mathcal{G}$ can be expressed as following equations :
\begin{equation}
\begin{array}{ll}
    LB(u)={{I}_{{{p}_{\min }}(u)}}(k,{{d}_{u}}-k+1)  \\
    UB(u)={{I}_{{{p}_{\max }}(u)}}(k,{{d}_{u}}-k+1)
\end{array}
\label{con:betafun}
\end{equation}
where ${{I}_{z}}(a,b)=\sum\nolimits_{i=a}^{a+b-1}\binom{a+b-1}{i}{{z}^{i}}{{(1-z)}^{a+b-1-i}}$ is the regularized beta function and can be computed in constant time using tables \cite{mckean1968tables}. Thus, we can compute the beta-function based bound in constant time.

Since there is usually a large difference between $p_{max}(u)$ and $p_{min}(u)$ in real-world uncertain graphs \cite{potamias2010k}, the beta-function based lower bound is not sufficiently tight. Next, we propose a more effective bound as follows.

\noindent
\paragraph{The Top-K Based Lower Bound.} Suppose the edges in $E(u,\mathcal{G})$ are sorted in non-increasing order based on the probabilities, i.e., ${{p}_{{{e}_{1}}}}\ge {{p}_{{{e}_{2}}}}\ge ...\ge {{p}_{{{e}_{{deg(u,\mathcal{G})}}}}}$. The set of top-k edges in $E(u,\mathcal{G})$ with the highest probabilities are denoted by $N_{u}^{k}(\mathcal{G})$, i.e., $N_{u}^{k}(\mathcal{G})=\{e_1,e_2,...,e_k\}$. Note that, $N_{u}^{k}(\mathcal{G})=\varnothing$ if $deg(u,\mathcal{G}) < k$. Below, we show the relationship between $k$-$prob(u,\mathcal{G})$ and $N_{u}^{k}(\mathcal{G})$.

\begin{lemma}
    Given an uncertain graph $\mathcal{G}$, an integer $k$, and a vertex $u \in \mathcal{G}$, we have $k$-$prob(u,\mathcal{G})$ $\ge \prod\nolimits_{e\in N_{u}^{k}(\mathcal{G})}{{{p}_{e}}}$.
\end{lemma}

\begin{proof}
    If $deg(u,\mathcal{G})<k$, $k$-$prob(u,\mathcal{G})$ $= \prod\nolimits_{e\in N_{u}^{k}(\mathcal{G})}{{{p}_{e}}}=0$. Otherwise, based on the possible world semantics, we use $\mathcal{G}_T$ to represent the set of all possible worlds where each possible world $ G \in \mathcal{G}_T$ contains all edges in $N_{u}^{k}(\mathcal{G})$. Thus, we have $\mathcal{G}_T\subseteq$$\ \mathcal{G}_{u}^{\ge k}$. According to Equation (\ref{equation:G}), $\Pr(\mathcal{G}_T)=\sum_{G\in \mathcal{G}_T}\Pr(G)=\prod\nolimits_{e\in N_{u}^{k}(\mathcal{G})}{{{p}_{e}}}$. We can infer that $\Pr (deg (u,\mathcal{G})\ge k) = \sum_{G\in \mathcal{G}_{u}^{\ge k}}\Pr(G)\ge \sum_{G\in \mathcal{G}_T}\Pr(G) = \prod\nolimits_{e\in N_{u}^{k}(\mathcal{G})}{{{p}_{e}}}$
\end{proof}

Since the neighbors of each vertex are sorted in descending order based on edge probability in the preprocessing stage, we take $O(k)$ to compute the Top-K based lower bound.

Based on the aforementioned methods, the bounds of $k$-$prob(u,\mathcal{G})$ can be updated as follows:
\begin{equation}
\resizebox{0.91\linewidth}{!}{%
$
\begin{array}{ll}
    LB(u)=\max\{{{I}_{{{p}_{\min }}(u)}}(k,{{d}_{u}}-k+1) , \prod\nolimits_{e\in N_{u}^{k}(\mathcal{G})}{{{p}_{e}}}\}  \\
    UB(u)={{I}_{{{p}_{\max }}(u)}}(k,{{d}_{u}}-k+1)
\end{array}
$
}
\label{eq:bounds}
\end{equation}

Then, we use the above bounds to estimate the $k$-probability of the vertices that need to be updated and reduce unnecessary refreshing.


\begin{definition}[Indefinite/Definite vertex]

    Given an uncertain graph $\mathcal{G}$, a vertex $u$ is an indefinite vertex if it satisfies that there exists a neighbor of $u$ that has been deleted and the $k$-probability of $u$ has not been recomputed; otherwise, $u$ is a definite vertex.
\end{definition}

In the peeling stage, given an integer $k$, we can calculate the $\eta$-threshold for all vertices by repeating the following steps:
\begin{itemize}
    \item \textbf{Step 1.} We obtain $p_{crt}\in \mathcal{G}$ and then delete $p_{crt}$ and all \textit{associative vertices} of $p_{crt}$.
    \item \textbf{Step 2.} We use the bounds in Equation (\ref{eq:bounds}) to estimate the $k$-probabilities of the \textit{indefinite vertices}. As shown in Figure \ref{fig:p_crt}, if $UB(u) \le k$-$prob(p_{crt},\mathcal{G})$, the $k$-probability of $u$ does not need to be recalculated and the $\eta$-$thres_k(u)$ is set to be $\eta$-$thres_k(p_{crt})$. Otherwise, $u$ may need to be recalculated. We add $u$ to a minimum heap $D$, where the key of $u$ is $LB(u)$.
    \item \textbf{Step 3.} We obtain a \textit{definite vertex} $p_{nxt}$ with the minimum $k$-probability among all \textit{definite vertices}.
    \item \textbf{Step 4.}  We compare $LB(u)$ for each vertex $u \in D$ with $k$-$prob(p_{nxt},\mathcal{G})$ to determine the necessity to recalculate $u$. As shown in Figure \ref{fig:p_nxt}, if $LB(u)\ge k$-$prob(p_{nxt},\mathcal{G})$, we can delay recalculating the $k$-probability of $u$; otherwise, $u$'s $k$-probability needs to be recalculated and we add $u$ to a set $X$.
    \item \textbf{Step 5.} We compute the $k$-probabilities for all vertices in $X$ and select the vertex with minimum $k$-probability in $X\cup \{p_{nxt}\}$ as $p_{crt}$.
\end{itemize}

To further reduce unnecessary computation in \textbf{Step 5}, we directly compute the $k$-probability of the vertex $u$ that satisfies $LB(u)<k$-$prob(p_{nxt},\mathcal{G})$ in \textbf{Step 4}, and dynamically update $p_{nxt}$ to the vertex with the smaller $k$-probability between $u$ and $p_{nxt}$. As a result, we can further decrease the number of vertices that need to be recalculated.

\noindent
\paragraph{Lazy Refreshing Strategy.} We refer to the above strategy as the lazy refreshing strategy. We use an on-demand refreshing approach to reduce the number of recomputations based on estimating $k$-probability. In \textbf{Step 2}, we use the upper bound to remove vertices that do not need to be recalculated. In \textbf{Step 4}, we delay the refreshing of vertices based on the lower bound and dynamically adjust the $k$-probability of $p_{nxt}$ to further avoid unnecessary recalculations. Example \ref{example:impUCF} further illustrates the effectiveness of our lazy refreshing strategy.

\begin{figure}[tbp]
\centering
\subfigure[Avoid refreshing based on the upper bound]{
    \label{fig:p_crt} 
    \includegraphics[width=0.22\textwidth]{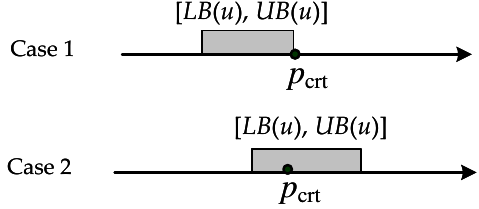}
}
\hfill
\subfigure[Delay refreshing based on the lower bound]{
    \label{fig:p_nxt} 
    \includegraphics[width=0.22\textwidth]{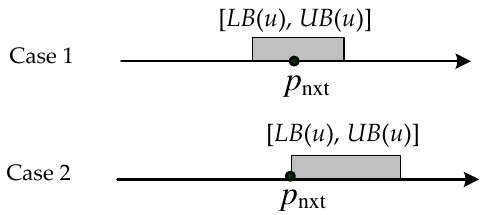}
}
\vspace{-1em}
\caption{Lazy refreshing strategy}
\vspace{-1em}
\label{fig:lazy refreshing}
\end{figure}

\begin{example}
    Considering the uncertain graph $\mathcal{G}$ in Figure \ref{fig:3n_core}, let $k=2$, we compute $k$-probabilities for $V\backslash \{v_{11}\}$ in initialization phase. In the next, we obtain the $p_{crt}=v_7$ in the first peeling. As $UB(v_3)=0.99954$, $UB(v_8)=0.36$, and $k$-$prob(v_7,\mathcal{G})=0.12$, $v_3,v_8$ are added into $D$. Then, we obtain $p_{nxt}=v_5$ and $k$-$prob(v_5,\mathcal{G})=0.28<LB(v_3)$ and $LB(v_8)$ in the second peeling. We obtain $p_{crt}=v_5$ and delete $v_5$ and its \textit{associative vertex} $v_6$. In the third peeling, the $p_{nxt}=v_{10}$ and $k$-$prob(v_{10},\mathcal{G})=0.35>LB(v_8)$. We recompute $k$-$prob(v_8,\mathcal{G})=0.30$ and obtain $p_{crt}=v_8$. Thus, $v_8$ and its associative vertices $v_9$ and $v_{10}$ are deleted. In the fourth peeling, we obtain $p_{nxt}=v_1$. As $k$-$prob(v_1,\mathcal{G})=0.80<LB(v_3)$, we obtain $p_{crt}=v_1$. Then, we delete $v_1$, and add $v_2,v_4$ into $D$ as $UB(v_2)=UB(v_4)=0.81$. Finally, $k\text{-}prob(v_2,\mathcal{G})=0.45$ and \textit{curThres}$=0.80$, we delete $v_2$ and its associative vertices $v_3$ and $v_4$. Compared to Baseline, the number of computations for $k$-$prob(v_3,\mathcal{G})$ is reduced from 4 to 0, which will significantly improve efficiency.
\label{example:impUCF}
\end{example}

\subsection{Progressive Refinement Strategy}
\label{subsec:floor}



To further improve the baseline algorithm, our strategy is to divide the vertices into different layers. Hence, the $k$-probabilities of vertices can be progressively refined layer by layer to search for the vertex with minimum $k$-probability.

We use $\Lambda_k=\{\eta$-$thres_k(u)|u\in (k,\eta)$-cores$\}$ to store all distinct $\eta$-thresholds and set $l=|\Lambda_k|$. Without loss of generality, $\Lambda_k=\{\eta_1,\eta_2,...,\eta_l\}$, and it is sorted in ascending order. We give the definition of partition of vertices:

\begin{definition}[Partition of Vertices]
    Given an uncertain graph $\mathcal{G}$ and three integers $k,i,l$, where $i \in [1,l]$, the partition of vertices in $\mathcal{G}$ is that $P(k,i)=\{u\in V|\eta$-$thres_{k+1}(u)=\eta_i\}$.
\label{definition:partition}
\end{definition}

The vertices partitioning framework is shown in Figure \ref{fig:vertex_partitioning}. Let $P(k)=\{P(k,1),P(k,2),...,P(k,l)\}$ and $PT(k)=V \backslash P(k)$, we can divide all vertices into $l+1$ layers based on their $\eta$-thresholds in the last iteration. Theorem \ref{theorem:qiantao} ensures the rationality of the vertices partitioning framework.


\begin{figure}[!tb]
    \centerline{\includegraphics[height=3.5cm]{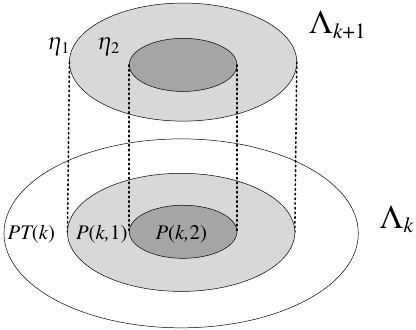}}
\caption{The vertices partitioning framework}
\label{fig:vertex_partitioning}
\vspace{-1em}
\end{figure}

\begin{theorem}
    Given an uncertain graph $\mathcal{G}$ and two parameter pairs of $(k_1,\eta_1)$ and $(k_2,\eta_2)$, where $k_1 \le k_2$ and $\eta_1 \le \eta_2$, then $(k_1,\eta_1)$-core and $(k_2,\eta_2)$-core are denoted by $H_1$ and $H_2$, respectively, and $H_2 \subseteq H_1$ holds.
\label{theorem:qiantao}
\end{theorem}

\begin{proof}
    For each $u \in H_2, \Pr(deg(u,\mathcal{G}) \ge k_2) \ge \eta_2$ based on Definition \ref{definition: kn_core}. As $k_1 \le k_2$ and $\eta_1 \le \eta_2, \Pr(deg(u,\mathcal{G}) \ge k_1) \ge \eta_2 \ge \eta_1$. As a result, each vertex of $H_2$ satisfies the constraint of $(k_1,\eta_1)$-core. Thus, $H_2 \subseteq H_1$ holds.
\label{lemma:yc}
\end{proof}

Given an integer $k$ and two probability thresholds $\eta_1,\eta_2$ with $\eta_1<\eta_2$. According to Theorem \ref{theorem:qiantao}, we have ($k+1,\eta_2$)-core $\subseteq$ ($k,\eta_2$)-core $\subseteq$ ($k,\eta_1$)-core. Then we can derive the following corollary, which gives us an \textbf{$\mathbf{\eta}$-threshold Based Lower Bound} to estimate the $k$-probability of a vertex.

\begin{corollary}
    Given an uncertain graph $\mathcal{G}$, an integer $k$, and a vertex $u \in \mathcal{G}$, we have $k$-$prob(u,\mathcal{G}) \ge \eta$-$thres_{k+1}(u)$.
\label{corollary:eta}
\end{corollary}

Therefore, once we have completed the partitioning of the vertices, we can improve the algorithm using following strategy called \textbf{progressive refinement}.

\begin{itemize}
    \item In the initialization phase for each $k$, we compute the $k$-probability for each vertex only in the first layer, i.e., $PT(k)$. Then, in order to obtain $p_{crt}$, the vertex with minimum $k$-probability, we first identify the vertex with minimum $k$-probability in $PT(k)$. This vertex is then assigned to $p_{crt}$. If the $k$-probability of $p_{crt}$ is no greater than the corresponding $\eta$-threshold of its adjacent higher layer, then the vertex with minimum $k$-probability is already found. Hence, $k$-probability initialization in higher layers can be saved. Otherwise, we need to refine the vertices’ $k$-probability in the next layer to find $p_{crt}$ until the $\eta$-threshold corresponding to the next layer is no less than $p_{crt}$’s $k$-probability. The subsequent process of searching for $p_{crt}$ is also implemented in this way. Thus, the initialization overhead and subsequent search space for $p_{crt}$ is significantly reduced.
\end{itemize}

After the vertex with minimum $k$-probability, $p_{crt}$, is deleted, bound estimation of its neighbors in higher layers can be enhanced by the $\eta$-threshold based lower bound. Specifically, suppose $p_{crt}$ is located at $PT(k)$, and $u$ is a neighbor of $p_{crt}$ that is located at $P(k,1)$. Then the lower bound of $u$’s $k$-probability can be further laid by the corresponding $\eta$-threshold of $P(k,1)$. Therefore, unnecessary refreshing can be saved. 

Note that we can complete vertices partitioning at a very low cost through the following method. For each $k$, we use one bucket to store the vertices in each $P(k,i)$ for $i \in [1,l]$. In the previous iteration, when the $\eta$-threshold of a vertex $u$ is obtained, we store $u$ in the corresponding bucket according to $u$'s $\eta$-threshold. This allows obtaining any $P(k,i)$ in $O(1)$ time in the current iteration.

\begin{algorithm}[!tb]
    \caption{refKPROB}
    \label{alg:refkprob}
    \textbf{Input}: An uncertain graph $\mathcal{G}(V,E,p)$, an integer $k$, a vertex heap $D$, and a vertex $p_{nxt}$\\
    \textbf{Output}: $p_{crt}$
    \begin{algorithmic}[1] 
        \WHILE{$D$ \textit{is not empty}}
        \STATE $u \leftarrow \arg \min_{u\in D}LB(u)$
        \STATE \textbf{if} $LB(u)\ge k$-$prob(p_{nxt},\mathcal{G})$
        \textbf{then break}
        \STATE Compute $k$-$prob(u,\mathcal{G})$ using DP
        \STATE \textbf{if} $k$-$prob(u,\mathcal{G})<\ $$k$-$prob(p_{nxt},\mathcal{G})$ \textbf{then} $p_{nxt} \leftarrow u$
        \STATE $D\leftarrow D \backslash \{u\}$
        \ENDWHILE
        \STATE \textbf{return} $p_{nxt}$
    \end{algorithmic}
\end{algorithm}

\subsection{Sketch of OptiUCF}
\label{subsec:opt}


The sketch of our proposed \textit{OptiUCF} algorithm is illustrated in Algorithms \ref{alg:refkprob}-\ref{alg:optimal_floor}. Algorithm \ref{alg:refkprob} presents the details for implementing refreshing to obtain $p_{crt}$. We check each vertex $u \in D$ that has minimum $LB(u)$ one by one. If vertex $u$ satisfies $LB(u)< k$-$prob(p_{nxt},\mathcal{G})$, we compute $k$-$prob(u,\mathcal{G})$ and progressively update $p_{nxt}$ to the vertex with smaller $k$-probability between $u$ and $p_{nxt}$ (lines 4-6). Otherwise, the algorithm comes to its end (line 3) and returns $p_{crt}$ (line 7).



\begin{algorithm}[!tb]
    \caption{OptiUCF}
    \label{alg:optimal_floor}
    \textbf{Input}: An uncertain graph $\mathcal{G}(V,E,p)$\\
    \textbf{Output}: Correct UCF-Index of $\mathcal{G}$
    \begin{algorithmic}[1] 
        \STATE $\Lambda_{k_{max}+1} \leftarrow \varnothing$ 
        \FOR{ $k \leftarrow k_{max}$ \textbf{to} 1}
        \STATE $\mathcal{G}'(V',E',p) \leftarrow$ the $k$-core of $\mathcal{G}$
        \STATE Obtain $\Lambda_{k+1}=\{\eta_1,\eta_2,...,\eta_l\}$
        \STATE Partition all vertices of $V'$ into $P(k)$ and $PT(k)$
        \STATE $V_w \leftarrow PT(k)$; $i\leftarrow 0$ \ \ $//$$V_w$ is a working set
        \STATE Compute $k$-$prob(u,\mathcal{G}')$ for all the vertices $u \in V_w$
        \STATE $S \leftarrow$ initialize an empty stack; \textit{curThres}$\leftarrow 0$
        \STATE $D \leftarrow \varnothing$
        \WHILE{$\mathcal{G}'$ is not empty}
        \STATE $p_{crt}\leftarrow \arg {{\min }_{u\in V_w \backslash D}}k$-$prob(u,\mathcal{G}')$
        \IF{$\exists u \in D \ \text{s.t.} \ LB(u)<k\text{-}prob(p_{crt},\mathcal{G}') $}
        \STATE \textbf{Call} \textit{refKPROB}$(\mathcal{G}', k, D, p_{crt})$ to obtain $p_{crt}$
        \ENDIF
        \WHILE{$\eta_{i+1}<k\text{-}prob(p_{crt},\mathcal{G}')$ \textbf{and} $i<l$}
        \STATE $i \leftarrow i+1$
        \STATE \textbf{Call} \textit{addLayer($V_w,\mathcal{G}',p_{crt},P(k),i$)} 
        \ENDWHILE
        \STATE \textit{curThres}$\leftarrow \max(k$-$prob(p_{crt},\mathcal{G}')$,\textit{curThres})
        \STATE $\eta$-$thres_k(p_{crt}) 
        \leftarrow$\textit{curThres}; $S$.push$(p_{crt})$
        \STATE $V_w \leftarrow V_w \backslash \{p_{crt}\}$; Remove $p_{crt}$ from $\mathcal{G}'$
        \FOR{ each $u \in N_{p_{crt}}(\mathcal{G}') \cap V_w$}
        \IF{$deg(u,\mathcal{G}')<k$ \textbf{or} $UB(u)\le$ \textit{curThres}} 
        \STATE $k$-$prob(u,\mathcal{G}') \leftarrow 0$
        \ELSE
        \STATE Add $u$ to $D$ \ \     $//$$UB(u)>$ \textit{curThres}
        \ENDIF
        \ENDFOR
        \IF{$V_w=\varnothing$ \textbf{and} $i<l$}
        \STATE $i \leftarrow i+1$
        \STATE \textbf{Call} \textit{addLayer}($V_w,\mathcal{G}',p_{crt},P(k),i$)
        \ENDIF
        \ENDWHILE
        \STATE Construct $\eta$-tree$_k$
        \ENDFOR
        \STATE \textbf{return} $\eta$-tree$_k$ for all $1\le k\le {{k}_{\max }}$
        \STATE \mbox{\textbf{Procedure} \textit{addLayer}($V_w,\mathcal{G}',p_{crt},P(k),i$)}
        \STATE \hspace{0.8em} $V_w\leftarrow V_w \cup P(k,i)$
        \hspace{1em}  
        \STATE \hspace{0.8em} \textbf{for} each $u\in P(k,i)$ \textbf{do}
        \STATE \hspace{1.8em} Compute $k$-$prob(u,\mathcal{G}')$ using DP
        \STATE \hspace{2.03em}\textbf{if} $k$-$prob(u,\mathcal{G}')<k$-$prob(p_{crt},\mathcal{G}')$ \textbf {then} $p_{crt}\leftarrow u$
    \end{algorithmic}
\end{algorithm}

The main program of \textit{OptiUCF} is presented in Algorithm \ref{alg:optimal_floor}. In the initialization phase (lines 3-9), vertices partitioning is implemented based on Definition \ref{definition:partition} and the $k$-probabilities of vertices in the outermost layer are computed. A minimum heap $D$ is set to store \textit{indefinite vertices}, and $V_w$ is a working set that is initialized by vertices in the outermost layer. In line 11, only the vertex in $V_w$ but not in $D$ is visited in search for the vertex with minimum $k$-probability, while \textit{indefinite vertices} in $D$ will be visited by delayed refreshing. If $V_w \backslash D$ is empty, a virtual vertex with $k$-probability larger than $1$ is returned. If there exists an \textit{indefinite vertex} in $D$ whose lower bound of $k$-probability is less than $p_{crt}$, delayed refreshing is triggered (lines 12-13). If the minimum $k$-probability found in the working set is larger than the corresponding $\eta$-threshold of the next layer, vertices in the next layer will be added into $V_w$, and the global minimum vertex will be searched over a larger region (lines 14-16). After deletion of $p_{crt}$ (line 19), its neighbors will be processed in lines 20-24.  Only the neighbor in the working set will be processed and its bound is laid in line 21. On the contrary, the neighbor outside the working set will not be processed. Because it must be located at a higher layer, its lower bound can be given by the corresponding $\eta$-threshold of its layer, and it will be processed by procedure \textit{addLayer} when necessary. It should be noted that in line 24, if $u$ has already been added in $D$, $u$ will be removed from $D$ first. Procedure \textit{addLayer} implements search space expansion and progressive refinement to find the global minimum vertex (lines 30-34).

\begin{theorem}
    Given an uncertain graph $\mathcal{G}(V,E,p)$, the time complexity of Algorithm \ref{alg:optimal_floor} is $O({{k}_{\max }}|V|^{2}\log |V|+\sum\nolimits_{u\in V}{k_{\max }^{2}\cdot deg {{(u,\mathcal{G})}^{2}})}+k_{\max }^{2}|E|)$.
\end{theorem}
    
\begin{proof}
    We calculate the $k$-core requires $O(|V|\text{+}|E|)$ time \cite{batagelj2003m}. In initialization phase, the time complexity of line 7 is $O(\sum\nolimits_{u\in V}{k\cdot deg (u,\mathcal{G}))}=O(k_{\max }|E|)$. In each peeling of lines 10-27, we use a Fibonacci heap to maintain all vertices in $V_w$. The Fibonacci heap takes a constant amount of time for insertion and updating if the key is decreased. In the worst-case scenario, we need to updated the $k$-probabilities of all vertices in $D$ during each peeling and totally take $O(|V|^2\log |V|)$ time to maintain $D$ and $V_w$ in peeling stage. Meanwhile, the $P(k)$ is added to $\mathcal{G}$ during obtaining $p_{crt}$ in the first peeling and every neighbor of $p_{crt}$ is recomputed. We need to take $ O(\sum\nolimits_{(u,v)\in E}{k\cdot deg (v,\mathcal{G})})=$$O(\sum\nolimits_{u\in V}{k\cdot deg {{(u,\mathcal{G})}^{2}})}$ to recompute the $k$-probability. The complexity of $\eta$-tree construction algorithm is $O(|E|)$ \cite{yang2019index}. As a result, the time complexity of Algorithm \ref{alg:optimal_floor} is $O({{k}_{\max }}|V|^{2}\log |V|+\sum\nolimits_{u\in V}{k_{\max }^{2}\cdot deg {{(u,\mathcal{G})}^{2}})}+k_{\max }^{2}|E|)$.
\end{proof}

\begin{table}[t]
    \centering
    \begin{tabular}{lcccc}
        \hline
        Datasets  & $|V|$ & $|E|$ & $d_{max}$ & $k_{max}$\\
        \hline
        Flickr     & 24,125     & 300,836  & 546 &225\\
        DBLP     & 684,911     & 2,284,991  & 611 &114\\
        Biomine  & 1,008,201    & 6,722,503  & 139,624 & 448\\
        WebGoogle & 875,713     & 5,105,039 & 6,332 & 44\\
        Youtube & 1,134,890     & 2,987,624  &28,754 &51\\
        WikiTalk & 2,394,385     & 5,021,410  &100,029 & 131\\
        Stackof & 2,601,977     & 63,497,050  &44,065 & 198\\
        LiveJournal & 4,847,571     & 68,993,773  &14,815 & 360\\
        \hline
    \end{tabular}
    \caption{Dataset statistics}
    \label{tab:plain}
\vspace{-1em}
\end{table}

\section{Experimental Evaluation}
\label{Section:experimental}
In this section, we conducted extensive experiments to evaluate the efficiency and scalability of the proposed algorithm.

\noindent
\paragraph{Experimental Setup.} We used three algorithms \textit{BC}, \textit{OP}, and \textit{OP*} to construct the correct UCF-Index. \textit{BC} denoted our basic algorithm in Section \ref{subsec:bc}, \textit{OP} denoted our \textit{OptiUCF} that only uses lazy refreshing strategy, and \textit{OP*} denoted the full version of \textit{OptiUCF}. These algorithms were implemented in C++ and compiled using g++ at the -O3 optimization level. All the experiments were conducted on a Linux server with Intel Xeon 2.50GHz CPU and 64GB RAM.

\noindent
\paragraph{Datasets.} We use eight real-world graphs to conduct the experiments. The summary of datasets is displayed in Table \ref{tab:plain}, where $d_{max}$ denotes the maximum degree and $k_{max}$ denotes the largest $k$-number. The edge probabilities in the first three datasets are derived from real-world applications \cite{potamias2010k}. Flickr is an online photo-sharing community, DBLP is a collaboration network, and Biomine is a biological interaction network. The details of the remaining five datasets can be found in SNAP \cite{snapnets}, and the edge probabilities are uniformly sampled from the interval [0,1].

\subsection{Efficiency}
Figure \ref{fig:total_time} shows the running time of \textit{BC}, \textit{OP}, and \textit{OP*} on each dataset. We can see that \textit{OP*} is the fastest across all datasets and outperforms \textit{BC} by 1 to 2 orders of magnitude. The effectiveness of our optimization strategies is validated. Compared to \textit{OP}, \textit{OP*} shows an improvement of 5 to 10 times across all datasets due to the effect of the progressive refinement technique. \textit{OP} performs better than \textit{BC}, which demonstrates the effectiveness of lazy refreshing strategy.


Three lower bounds, namely $\eta$-threshold based, Top-K based, and Beta-Function based lower bound are proposed in this paper. The $\eta$-threshold based lower bound is adopted to support progressive refinement. We hereby verify the impact of the rest two lower bounds on the overall performance of our algorithm. During the experiment, only one lower bound is used while the other is masked, and the results are shown in Figure \ref{fig:lowerbound_time}. \textit{OP*-Topk} indicates using single Top-K based lower bound in \textit{OptiUCF}, and \textit{OP*-Beta} means using single Beta-Function based lower bound. We observe that \textit{OP*-Topk} is 10\% to 25\% faster than \textit{OP*-Beta}, which demonstrates that Beta-Function based lower bound is less tight than Top-K based lower bound.


\subsection{Accuracy}

\begin{figure}[tb]
    \centerline{\includegraphics[width=0.4\textwidth]{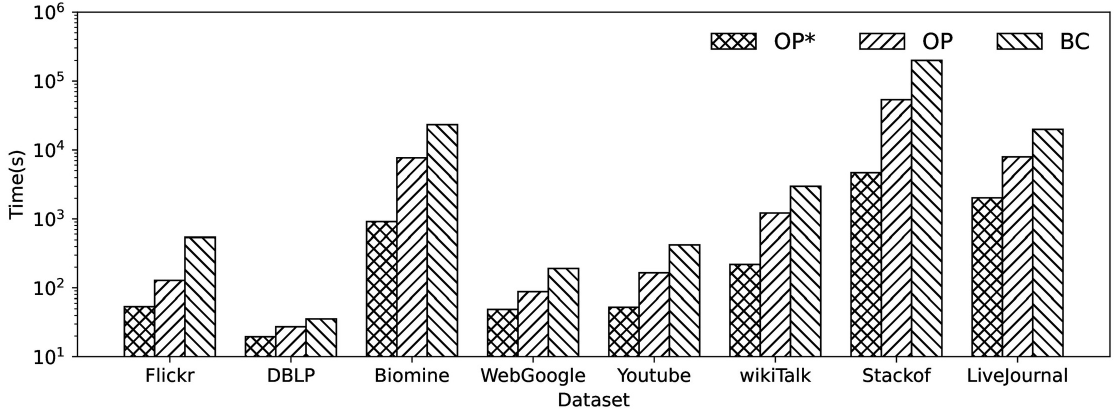}}
    \vspace{-0.225em}
\caption{Time cost for UCF-Index construction}
\label{fig:total_time}
\vspace{-0.285em}
\end{figure}

\begin{figure}[!t]
    \centerline{\includegraphics[width=0.4\textwidth]{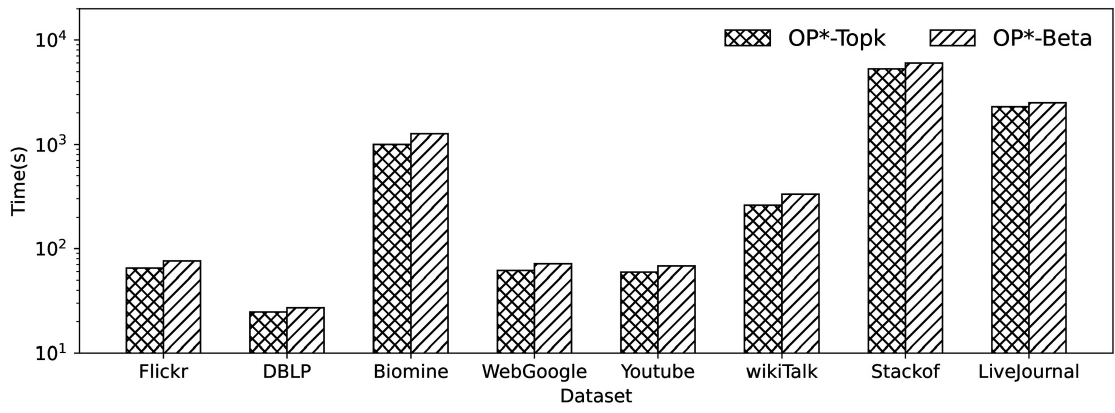}}
    \vspace{-0.235em}
\caption{Time cost for UCF-Index with different lower bounds}
\label{fig:lowerbound_time}
\vspace{-1em}
\end{figure}

The UCF-Index construction algorithm denoted as \textit{EC} \cite{yang2019index} uses floating-point division to update $k$-probabilities, which will result in significant errors in the calculated $\eta$-thresholds. The proportion of vertices with incorrectly calculated $\eta$-threshold on Biomine and LiveJournal is shown in Figure \ref{accuracy}, and the results on the other graphs show consistent trends. We can clearly see that as $k$ increases, the error ratio increases sharply and approaches 100\% due to the increase in the number of divisions used. These results indicate that the \textit{EC} algorithm is extremely inaccurate for $\eta$-thresholds, hence the correctness of the index so constructed cannot be guaranteed.

\subsection{Scalability}
In this experiment, we evaluate the scalability of our proposed algorithms using two real-world graphs, Biomine and LiveJournal. Similar trend can be found in the other graphs. The number of vertices randomly sampled from the original graph is varying from 20\% to 100\%. 

The experimental results are shown in Figure \ref{scalable}. As the graph size increases, the running time of our algorithm increases smoothly, and the gap between \textit{OP*} and \textit{BC} gradually widens. These indicate that our optimization techniques are of good scalability and are more efficient on large graphs.

\section{Related Works}
Dense subgraph mining on uncertain graphs has been extensively studied \cite{zhang2024finding,cheng2021efficient,qiao2021maximal,dai2022fast}. Bonchi et al. first proposed the $k$-core search problem on uncertain graphs, namely $(k,\eta)$-core \cite{bonchi2014core}. Li et al. designed a new dynamic programming algorithm to accelerate the computation of $\eta$-degree \cite{li2019improved}. Dai et al. addressed the issues arising from floating-point number division calculation errors during the $\eta$-degree update process \cite{dai2021core}. Yang et al. introduced a forest-based index structure called UCF-Index, which can support query of $(k,\eta)$-cores in an optimal time for any input parameters $k$ and $\eta$ \cite{yang2019index}. However, the UCF-Index still faces a issue with floating-point calculation errors.

Sun et al. designed a CPT-Index to support online queries for $(k,\gamma)$-truss and proposed a top-down index construction algorithm \cite{sun2021efficient}. Xing et al. solved the problem of floating-point division errors in the computation of $(k,\gamma)$-truss and proposed an accurate UTF-Index to support the search of the connected $(k,\gamma)$-truss community \cite{xing2024truss}. These studies on $(k,\gamma)$-truss inspire our work. Nevertheless, the key distinction is that $(k,\gamma)$-truss focuses on the edge support probability, while $(k,\eta)$-core focuses on vertex's $\eta$-degree. Therefore, our work is clearly different from theirs.


\begin{figure}[!t]
\subfigcapskip=0pt
\centering
\subfigure[Biomine]{
    \label{error_rate_Biomine} 
    \includegraphics[width=0.222\textwidth]{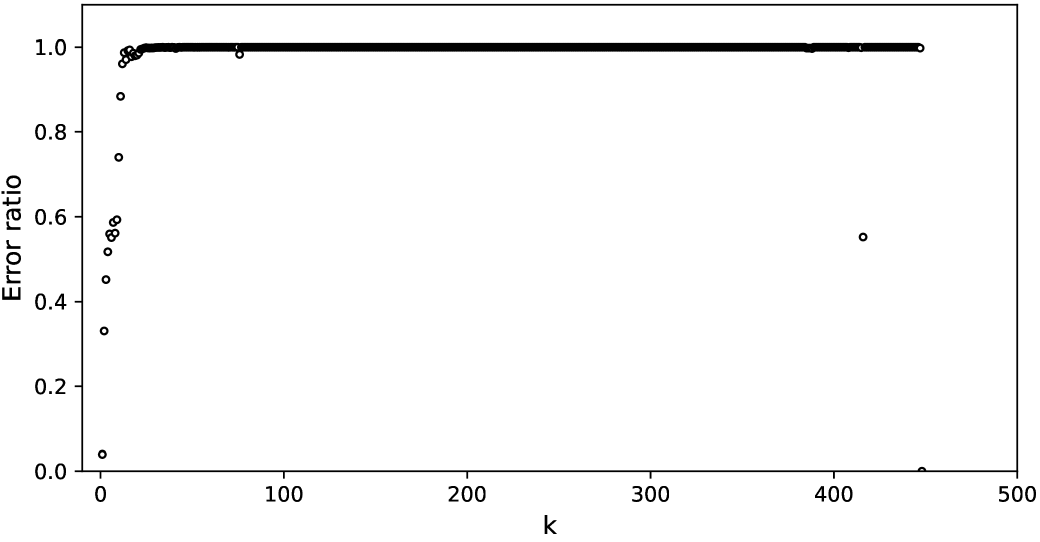}
}
\subfigure[LiveJournal]{
    \label{error_rate_LiveJournal} 
    \includegraphics[width=0.222\textwidth]{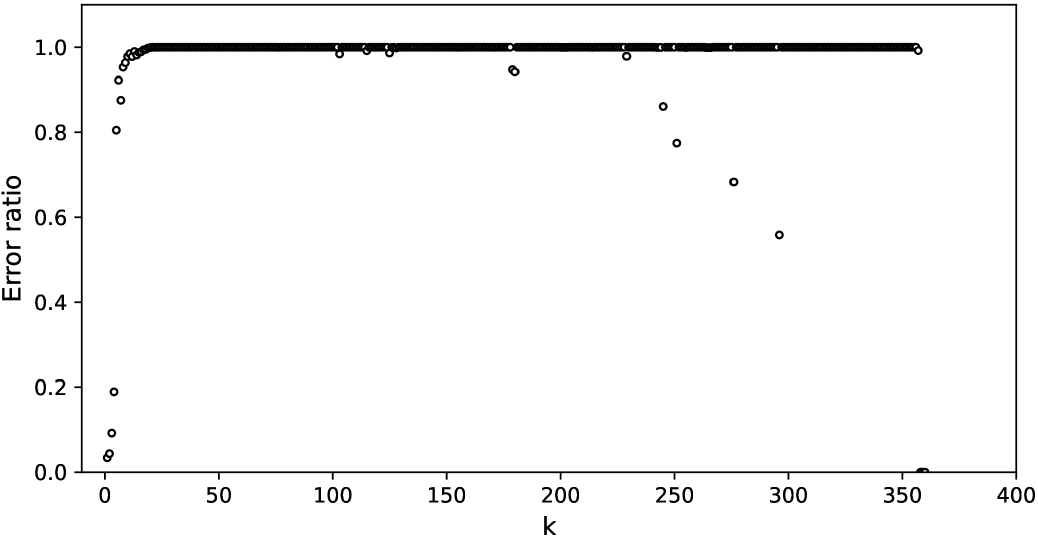}
}
\vspace{-0.65em}
\caption{The error ratio of $\eta$-threshold computed by \textit{EC}}
\label{accuracy}
\vspace{-1em}
\end{figure}

\begin{figure}[!t]
\subfigcapskip=0pt
\centering
\subfigure[Biomine]{
    \label{scalable_Biomine} 
    \includegraphics[width=0.222\textwidth]{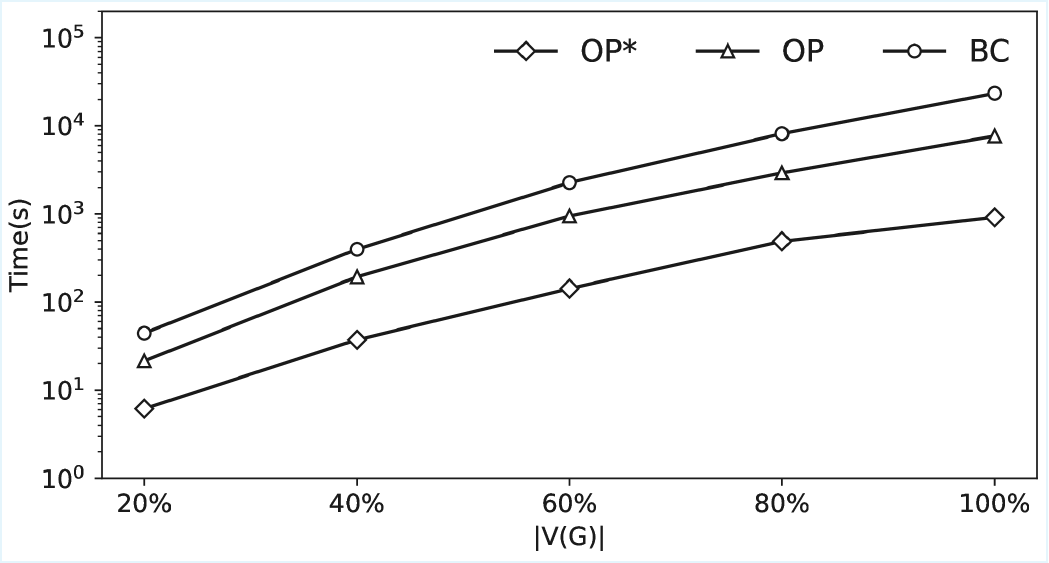}
}
\subfigure[LiveJournal]{
    \label{scalable_LiveJournal} 
    \includegraphics[width=0.222\textwidth]{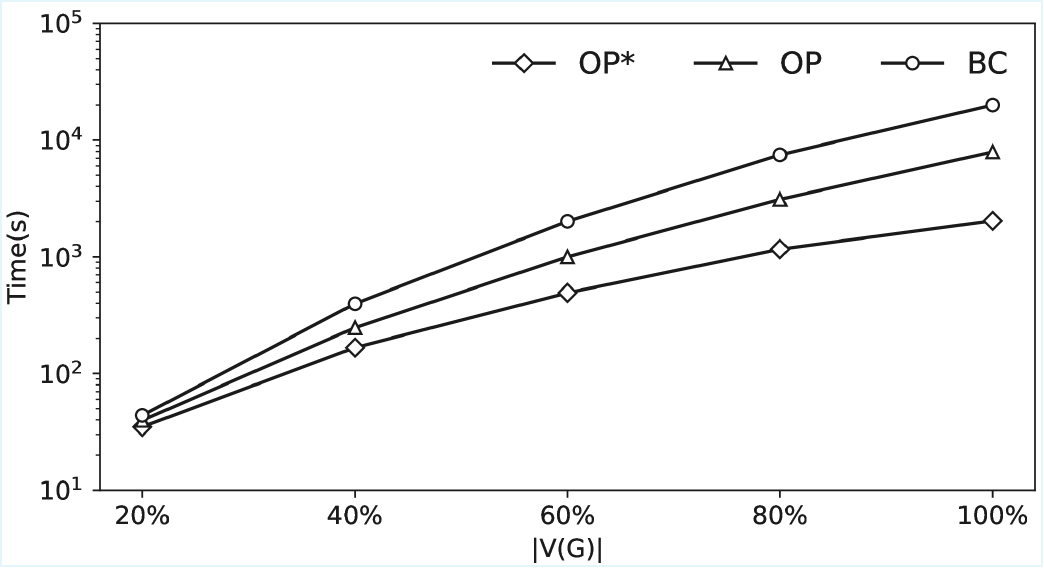}
}
\vspace{-0.65em}
\caption{Scalability of UCF-Index construction}
\label{scalable}
\vspace{-0.875em}
\end{figure}

\section{Conclusion}
In this paper, we propose an efficient algorithm to build a correct UCF-Index. Our basic strategy is to recalculate the $k$-probabilities that need to be updated. We propose a lazy refreshing strategy to minimize the number of recalculations, and develop a progressive refinement to reduce the overhead to search for vertex with minimum $k$-probability. Our proposed algorithm outperforms the baseline algorithm by one to two orders of magnitude.


\bibliographystyle{named}
\bibliography{ijcai25}

\end{document}